\documentclass[a4paper,USenglish,cleveref, autoref, thm-restate ]{lipics-v2021}

\usepackage[utf8]{inputenc}
\usepackage{amsmath}     \usepackage{amssymb}     \usepackage{mathtools}   \usepackage{microtype}   \usepackage{multirow}    \usepackage{booktabs}    \usepackage{subcaption}  \usepackage[ruled,linesnumbered]{algorithm2e}  \usepackage[usenames,dvipsnames,table]{xcolor}  \usepackage{tikz}
\usetikzlibrary{shapes.geometric,arrows.meta,arrows}
\usetikzlibrary{calc}
\usepackage{nag}        \usepackage{todonotes}  \usepackage{tcolorbox}  \usepackage{floatrow}
\usepackage{hyperref}   

\hideLIPIcs

\newcommand{\sump}{\sqcup}

\newcommand{\cc}[1]{{\mbox{\textnormal{\textsf{#1}}}}\xspace} 
\newcommand{\NP}{\cc{NP}}

\newcommand{\XP}{\cc{XP}}

\newcommand{\ang}[1]{\langle #1 \rangle}
\newcommand{\sep}{\;|\;}
\newcommand{\seq}{\subseteq}

\newcommand{\ca}[1]{\mathcal{#1}}
\newcommand{\bigoh}{\mathcal{O}}
\newcommand{\tww}{\textnormal{tww}}

\usepackage{pgfplots}
\pgfplotsset{compat=1.15}
\usepackage{mathrsfs}
\usepackage{tikz}
\usetikzlibrary{arrows}
\usetikzlibrary{backgrounds}
\usetikzlibrary{hobby}
\usetikzlibrary{decorations.markings}
\usetikzlibrary{math}

\usepackage{nameref,cleveref}
\Crefname{splemma}{Lemma}{Lemmas}
\Crefname{sptheorem}{Theorem}{Theorems}
\Crefname{spdefinition}{Definition}{Definitions}
\Crefname{spproperty}{Property}{Properties}
\Crefname{spcorollary}{Corollary}{Corollaries}

\author{Jakub Balab\' an}{Faculty of Informatics, Masaryk University, Brno, Czech Republic}{485053@mail.muni.cz}{0000-0002-2475-8938}{}

\author{Robert Ganian}{Algorithms and Complexity Group, TU Wien, Vienna, Austria}{rganian@gmail.com}{0000-0002-7762-8045}{Robert Ganian acknowledges support by the FWF and WWTF Science Funds (FWF project 10.55776/Y1329 and WWTF project ICT22-029).}

\author{Mathis Rocton}{Algorithms and Complexity Group, TU Wien, Vienna, Austria}{mrocton@ac.tuwien.ac.at}{0000-0002-7158-9022}{Mathis Rocton acknowledges support by the \includegraphics[width=0.5cm]{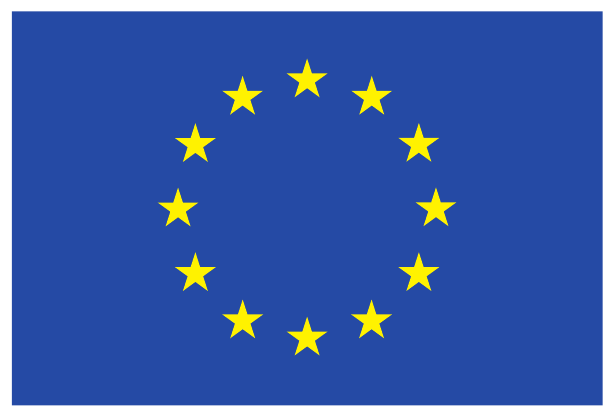} European Union's Horizon 2020 research and innovation COFUND programme (LogiCS@TUWien, grant agreement No 101034440), and the FWF Science Fund (FWF project 10.55776/Y1329).}

\title{Twin-Width Meets Feedback Edges\\ and Vertex Integrity}

\titlerunning{Twin-Width Meets Feedback Edges and Vertex Integrity}
\ccsdesc[500]{Theory of computation~Parameterized complexity and exact algorithms}
\keywords{twin-width, fixed-parameter algorithms, feedback edge number, vertex integrity}

\nolinenumbers

\begin{document}
\maketitle

\begin{abstract}
The approximate computation of twin-width has attracted significant attention already since the moment the parameter was introduced. A recently proposed approach (STACS 2024) towards obtaining a better understanding of this question is to consider the approximability of twin-width via fixed-parameter algorithms whose running time depends not on twin-width itself, but rather on parameters which impose stronger restrictions on the input graph. The first step that article made in this direction is to establish the fixed-parameter approximability of twin-width (with an additive error of 1) when the runtime parameter is the feedback edge number.

Here, we make several new steps in this research direction and obtain:

\begin{itemize}
\item An asymptotically tight bound between twin-width and the feedback edge number;

\item A significantly improved fixed-parameter approximation algorithm for twin-width under the same runtime parameter (i.e., the feedback edge number) which circumvents many of the technicalities of the original result and simultaneously avoids its formerly non-elementary runtime dependency;

\item An entirely new fixed-parameter approximation algorithm for twin-width when the runtime parameter is the vertex integrity of the graph.
\end{itemize}
\end{abstract}

\section{Introduction}
\emph{Twin-width} is a comparatively recent graph-theoretic measure which is the culmination of as well as a catalyst for several recent breakthroughs in the area of algorithmic model theory~\cite{DBLP:conf/soda/BonnetGKTW21,DBLP:conf/icalp/BonnetG0TW21,DBLP:conf/stoc/BonnetGMSTT22,DBLP:conf/stacs/BonnetGMT23,DBLP:conf/soda/BonnetKRT22}.
Indeed, it has the potential to provide a unified explanation of why model-checking first order logic is fixed-parameter tractable on a number of graph classes which were, up to then, considered to be separate islands of tractability for the model-checking problem. This includes graphs of bounded rank-width, proper minor-closed graphs, map graphs~\cite{BonnetKTW22}, bounded-width posets~\cite{DBLP:conf/iwpec/BalabanH21} as well as a number of other specialized graph classes~\cite{DBLP:conf/wg/BalabanHJ22,EppsteinConfluent}.

While twin-width is related to graph parameters such as rank-width and path-width~\cite{DBLP:conf/soda/BonnetKRT22} 
as well as to measures which occur in matrix theory such as excluding linear minors~\cite{DBLP:conf/stacs/BonnetGMT23}, what distinguishes twin-width from these other measures is that we lack efficient algorithms for computing the twin-width of a graph. In particular, it is known that already deciding whether a graph has twin-width at most $4$ is \NP-hard~\cite{BergeBD22}. This is highly problematic for the following reason: virtually every known algorithm that uses twin-width requires access to a so-called \emph{contraction sequence}, which serves the same role as the decompositions typically used for classical parameters such as treewidth~\cite{RobertsonS83} and rank-width~\cite{Oum05}. Intuitively speaking, a contraction sequence of width $t$---which serves as a witness for $G$ having twin-width at most $t$
---of a graph $G$ is a sequence $C$ of contractions of (not necessarily pairwise adjacent) vertex pairs which satisfies the following property: at each step of $C$, every vertex $v$ only has at most $t$ neighbors with an ancestor that is not adjacent to some ancestor of $v$\footnote{Formal definitions are provided in Section~\ref{sec:prelims}.}. 

The aforementioned \NP-hardness of identifying graphs of twin-width $4$~\cite{BergeBD22} effectively rules out fixed-parameter as well as \XP\ algorithms for computing optimal contraction sequences when parameterized by the twin-width itself. One possible approach to circumvent this obstacle would be to devise a fixed-parameter algorithm which still uses the twin-width $t$ as the parameter and computes at least an approximately-optimal contraction sequences, i.e., a contraction sequence of width $f(t)$ for some computable function $f$. On a complexity-theoretic level, such a result may be seen as ``almost'' as good as computing twin-width exactly, as it would still yield a fixed-parameter algorithm for first-order model checking.

Unfortunately, the task of finding such an algorithm has proven to be highly elusive, and it is far from clear that one even exists---in fact, whether twin-width can be approximated in fixed-parameter time (for any function $f$ of the twin-width) can be seen as arguably the most prominent open question in contemporary research of twin-width. A recent approach towards attacking this question proposed by Balab\' an, Ganian and Rocton~\cite{BalabanGR24} is to first relax the running time requirement and ask whether we can obtain an $f(t)$-approximation for twin-width at least via a fixed-parameter algorithm where the runtime parameter is different (and, in particular, larger) than the twin-width $t$ itself. As a first step in this direction, Balab\' an, Ganian and Rocton developed a highly non-trivial fixed-parameter algorithm that computes a contraction sequence of width at most $t+1$ and is parameterized by the \emph{feedback edge number} of the input graph, i.e., the edge deletion distance to acyclicity. In the same paper, they also showed that the twin-width of a graph with feedback edge number $k$ is upper-bounded by $k+1$.

\subparagraph*{Contributions.}
In this article, we significantly expand on the results of Balab\' an, Ganian and Rocton~\cite{BalabanGR24} and present the next steps in the overarching program of understanding the boundaries of tractability for computing approximately-optimal contraction sequences. We summarize our three main contributions below.

In Section~\ref{sec:bound}, we revisit the relationship between twin-width and the feedback edge number. Here, we improve the previous linear bound of Balab\' an, Ganian and Rocton~\cite{BalabanGR24} to square-root,
and also show that our new bound is asymptotically tight. More precisely, we show that every graph class with feedback edge number $k$ has twin-width $\bigoh(\sqrt{k})$ (Theorem~\ref{thm:upper}), and also construct a graph class with feedback edge number $k$ whose twin-width is lower-bounded by $\Theta(\sqrt{k})$ (Proposition~\ref{prop:lower}).

In Section~\ref{sec:improvedfen}, we revisit the main result of Balab\' an, Ganian and Rocton~\cite{BalabanGR24}: a polynomial-time reduction procedure which transforms every input graph $G$ with feedback edge number $k$ and twin-width $t$ into a (tri-)graph $G'$ whose twin-width lies between $t$ and $t+1$ and whose size is upper-bounded by a non-elementary function of $k$. While this suffices to obtain the desired fixed-parameter approximation algorithm (as one may brute-force over all contraction sequences of $G'$), the dependence on the parameter $k$ is astronomical and the proof relies on a sequence of highly technical arguments about how a hypothetical contraction sequence may be retrofitted in order to avoid certain degenerate steps. As our second main contribution, we provide a new proof for the fixed-parameter approximability of twin-width parameterized by the feedback edge number which not only avoids many of the technical difficulties faced in the previous approach, but crucially also improves the size bound for the reduced instance $G'$ from a \emph{non-elementary} to a \emph{quadratic} function of $k$.

Finally, in Section~\ref{sec:vi} we push the frontiers of approximability for twin-width by obtaining an algorithm which computes a contraction sequence for $G$ of width at most twice the graph's twin-width and runs in time 
$f(p)\cdot|G|$, where $p$ is the \emph{vertex integrity} of $G$. Vertex integrity is a parameter which intuitively measures how easily a graph may be separated into small parts, and is defined as the smallest integer $p$ such that there exists a separator $X$ with the following property: each connected component $C$ of $G-X$ satisfies $|C\cup X|\leq p$.
Vertex integrity may be seen as the natural intermediate step between the \emph{vertex cover number} (which is the size of the smallest vertex cover in $G$, and which is known to allow for a trivial fixed-parameter algorithm for computing twin-width) and decompositional parameters such as \emph{treedepth} and \emph{treewidth} (for which the existence of a fixed-parameter approximation algorithm for twin-width remains a prominent open question~\cite{BalabanGR24}).
Our result relies on a data reduction procedure which incorporates entirely different arguments than those used for the feedback edge number, and the correctness proof essentially shows that every optimal contraction sequence can be transformed into a near-optimal one where all ``similar parts'' of $G$ are treated in a ``similar way''.

\subparagraph*{Related Work.}
Beyond the setting of computing twin-width and the associated contraction sequences, there are numerous other works which have targeted fixed-parameter algorithms for computing a structural graph parameter $X$ when parameterized by graph parameters that differ from $X$. The general aim in this research direction is typically to further one's understanding of the fundamental problem of computing the targeted parameter $X$. Examples of fixed-parameter algorithms obtained in this setting include those for treewidth parameterized by the feedback vertex number~\cite{DBLP:journals/siamdm/BodlaenderJK13},
treedepth parameterized by the vertex cover number~\cite{DBLP:conf/iwpec/KobayashiT16}, 
MIM-width parameterized by the feedback edge number and other parameters~\cite{DBLP:conf/innovations/EibenGHJK22}, 
and the directed feedback vertex number parameterized by the (undirected) feedback vertex number~\cite{DBLP:journals/algorithmica/BergougnouxEGOR21}. The feedback edge number and vertex integrity have also been used to obtain parameterized algorithms for a number of other challenging problems~\cite{UhlmannW13,BannisterCE18,GanianO21,LampisM21,GanianK21,GimaHKKO22,FichteGHSO23,GimaO24}, whereas the latter parameter has also been studied in the literature under other asymptotically-equivalent names such as the \emph{fracture number}~\cite{DvorakEGKO21,GanianKO21} and \emph{starwidth}~\cite{Ee17}. We refer interested readers to the very recent manuscript of Hanaka, Lampis, Vasilakis and Yoshiwatari~\cite{VIarxiv} for a more detailed overview of vertex integrity and its relationship to other fundamental graph measures.

\section{Preliminaries}
\label{sec:prelims}

For integers $i$ and $j$, we let $[i,j] := \{n \in \mathbb N \sep i \le n \le j\}$ and $[i] := [1, i]$. 
We assume familiarity with basic concepts in graph theory~\cite{Diestel} and parameterized algorithmics~\cite{DowneyF13,CyganFKLMPPS15}.
When $H$ is an induced subgraph of $G$, we denote it by $H \seq G$.
Given vertex sets $X$ and $U$, we will use $G[X]$ to denote the graph induced on $X$ and $G - U$ to denote the graph $G[V(G)\setminus U]$; similarly, for an edge set $F$, $G-F$ denotes $G$ after removing the edges in $F$. 

A \emph{dangling path} in $G$ is a path of vertices which all have degree $2$ in $G$, and a \emph{dangling tree} in $G$ is an induced subtree in $G$ which can be separated from the rest of $G$ by removing a single edge.
The \emph{length} of a path is the number of edges it contains.
The \emph{distance} between two vertices $u$ and $v$ 
is the length of the shortest path between them.

An edge set $F$ in an $n$-vertex graph $G$ is called a \emph{feedback edge set} if $G-F$ is acyclic, and the \emph{feedback edge number} of $G$ is the size of a minimum feedback edge set in $G$.
 We remark that a minimum feedback edge set can be computed in time $\bigoh(n)$ as an immediate corollary of the classical (DFS- and BFS-based) algorithms for computing a spanning tree in an unweighted graph $G$.

A graph has \emph{vertex integrity} $p$ if $p$ is the smallest integer with the following property: $G$ contains a vertex set $S$ such that for each connected component $H$ of $G-S$, $|V(H)\cup S|\leq p$. One may observe that 
the vertex integrity is upper-bounded by the size of a minimum vertex cover in the graph (i.e., the vertex cover number) plus one, 
and both vertex integrity and the feedback edge number are lower-bounded by treewidth minus one~\cite{RobertsonS83}.
The vertex integrity of an $n$-vertex graph can be computed in time $\bigoh(p^{p+1}\cdot n)$~\cite{DrangeDH16}.

\smallskip
\noindent \textbf{Twin-Width.}\quad
In order to provide a seamless transition between our new results and the previous algorithm of Balab\' an, Ganian and Rocton~\cite{BalabanGR24} (which our Section~\ref{sec:improvedfen} improves and builds on), we will closely follow the terminology introduced in their preceding work.

A \emph{trigraph} $G$ is a graph whose edge set is partitioned into a set of \emph{black} and \emph{red} edges. The set of red edges is denoted $R(G)$, and the set of black edges $E(G)$.
The \emph{black (\emph{resp}.\ red) degree} of $u\in V(G)$ is the number of black (resp.\ red) edges incident to $u$ in $G$.
We extend graph-theoretic terminology to trigraphs by ignoring the colors of edges; for example, the degree of $u$ in $G$ is the sum of its black and red degrees (in the literature, this is sometimes called the \emph{total degree}).
We say a (sub)graph is \emph{black (\emph{resp}.\ red)} if all of its edges are black (resp. red); for example, $P$ is a red path in $G$ if it is a path containing only red edges. Without a color adjective, the path (or a different kind of subgraph) may contain edges of both colors. We use $G[Q]$ to denote the subgraph of $G$ induced on $Q\subseteq V(G)$.

Given a trigraph $G$, a \emph{contraction} of two distinct vertices $u,v\in V(G)$ is the operation which produces a new trigraph by (1) removing $u, v$ and adding a new vertex $w$, (2) adding a black edge $wx$ for each $x\in V(G)$ such that $xu$, $xv\in E(G)$, and (3) adding a red edge $wy$ for each $y\in V(G)$ such that $yu\in R(G)$, or $yv\in R(G)$, or $y$ contains only a single black edge to either $v$ or $u$.
A sequence $C = (G = G_1,\ldots,G_n)$ is a \emph{partial contraction sequence of $G$} if it is a sequence of trigraphs such that for all $i\in [n-1]$, $G_{i+1}$ is obtained from $G_i$ by contracting two vertices.
A \emph{contraction sequence} is a partial contraction sequence which ends with a single-vertex graph.
The \emph{width} of a (partial) contraction sequence $C$, denoted $w(C)$, is the maximum red degree over all vertices in all trigraphs in $C$.
The \emph{twin-width} of $G$, denoted $\tww(G)$, is the minimum width of any contraction sequence of $G$, and a contraction sequence of width $\tww(G)$ is called \emph{optimal}. An example of a contraction sequence is provided in Figure~\ref{fig:seq}.

\begin{figure}[ht]
\begin{tikzpicture}[line cap=round,line join=round,>=triangle 45,x=1.0cm,y=1.0cm]

\tikzset{
    vertex/.style = {draw, circle, fill=gray, minimum width=4pt, inner sep=0pt}}
    
\begin{scriptsize}
\node[vertex] (a) {};
\node[vertex] (b) [right of = a] {};
\node[vertex] (c) [above of = a] {};
\node[vertex] (d) [right of = c] {};
\node[vertex] (e) [above of = c] {};
\node[vertex] (f) [right of = e] {};
\draw (a)--(b)--(c)--(f)--(e)--(d)--(b);
\draw(a)--(c)--(e);

\node () at (a) [left=2pt] {$A$};
\node () at (c) [left=2pt] {$C$};
\node () at (e) [left=2pt] {$E$};
\node () at (b) [right=2pt] {$B$};
\node () at (d) [right=2pt] {$D$};
\node () at (f) [right=2pt] {$F$};

\node[vertex] (a2) [right = 40pt of b]{};
\node[vertex] (b2) [right of = a2] {};
\node[vertex] (c2) [above of = a2] {};
\node[vertex] (d2) [right of = c2] {};
\node[vertex] (ef) [above of = c2] {};
\draw (a2)--(b2)--(c2)--(ef)--(d2)--(b2);
\draw(a2)--(c2);
\draw[color=red, thick] (d2)--(ef);

\node () at (a2) [left=2pt] {$A$};
\node () at (c2) [left=2pt] {$C$};
\node () at (b2) [right=2pt] {$B$};
\node () at (d2) [right=2pt] {$D$};
\node () at (ef) [right=2pt] {$EF$};

\node[vertex] (ab) [right = 40pt of b2]{};
\node[vertex] (c3) [above of = ab] {};
\node[vertex] (d3) [right of = c3] {};
\node[vertex] (ef3) [above of = c3] {};
\draw (ab)--(c3)--(ef3);
\draw[color=red, thick] (ab)--(d3)--(ef3);

\node () at (ab) [right=2pt] {$AB$};
\node () at (c3) [left=2pt] {$C$};
\node () at (d3) [left=3pt] {$D$};
\node () at (ef3) [right=2pt] {$EF$};

\node[vertex] (cd) [right = 40pt of d3] {};
\node[vertex] (ab4) [below of= cd]{};
\node[vertex] (ef4) [above of = cd] {};
\draw[color=red, thick] (ab4)--(cd)--(ef4);
\node () at (ab4) [left=2pt] {$AB$};
\node () at (cd) [left=2pt] {$CD$};
\node () at (ef4) [left=2pt] {$EF$};

\node[vertex] (cdef) [right = 40pt of cd] {};
\node[vertex] (ab5) [below of= cdef]{};
\draw[color=red, thick] (ab5)--(cdef);
\node () at (ab5) [left=2pt] {$AB$};
\node () at (cdef) [above = 2pt] {$CDEF$};

\node[vertex] (fin) [right = 40pt of ab5] {};
\node () at (fin) [above = 2pt] {$ABCDEF$};

\end{scriptsize}
\end{tikzpicture}
\caption{A contraction sequence of width 2 for the leftmost graph, consisting of $6$ trigraphs.
\label{fig:seq}}
\end{figure}
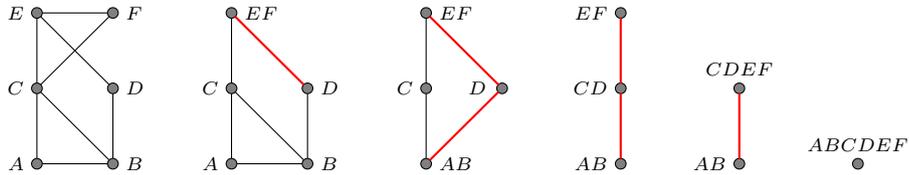

Let us now fix a contraction sequence $C = (G = G_1,\ldots,G_n)$.
For each $i \in [n]$, we associate each vertex $u \in V(G_i)$ with a set $\beta(u, i) \seq V(G)$, called the \emph{bag} of $u$, which contains all vertices contracted into $u$.
Formally, we define the bags as follows:

\begin{itemize}
\item for each $u \in V(G)$, $\beta(u, 1) := \{u\}$;
\item for $i \in [n-1]$, if $w$ is the new vertex in $G_{i+1}$ obtained by contracting $u$ and $v$, then $\beta(w, i+1) := \beta(u, i) \cup \beta(v,i)$; otherwise, $\beta(w, i+1) := \beta(w, i)$. 
\end{itemize}

Note that if a vertex $u$ appears in multiple trigraphs in $C$, then its bag is the same in all of them, and so we may denote the bag of $u$ simply by $\beta(u)$.
Let us fix $i, j \in [n]$, $i \le j$.
If $u \in V(G_i)$, $v \in V(G_j)$, and $\beta(u) \seq \beta(v)$, then we say that $u$ is \emph{an ancestor} of $v$ in $G_i$ and $v$ is \emph{the descendant} of $u$ in $G_j$ (clearly, this descendant is unique).
If $H$ is an induced subtrigraph of $G_i$, then $u \in V(G_j)$ is a \emph{descendant} of $H$ if it is a descendant of at least one vertex of $H$. 
A contraction of $u, v \in V(G_j)$ into $uv \in V(G_{j+1})$ \emph{involves} $w \in V(G_i)$ if $w$ is an ancestor of $uv$.

The following definition provides terminology that allows us to partition a contraction sequence into ``steps'' based on contractions between a subset of vertices in the original graph.

\begin{definition}\label{def:restriction}
Let $C$ be a contraction sequence of a trigraph $G$, and let $H$ be an induced subtrigraph of $G$ with $|V(H)| = m$.
For $i \in [m-1]$, let $C\ang i_H$ be the trigraph in $C$ obtained by the $i$-th contraction between two descendants of $H$, and let $C\ang 0_H = G$. For $i\in [m-1]$, let $U_i$ and $W_i$ be the bags of the vertices which are contracted into the new vertex of $C\ang i_H$.

A contraction sequence $C[H] = (H = H_1, \ldots, H_m)$ is the \emph{restriction of $C$ to $H$} if for each $i \in [m-1]$, $H_{i+1}$ is obtained from $H_{i}$ by contracting the two vertices $u, w \in V(H_{i})$ such that $\beta(u) = U_i \cap V(H)$ and $\beta(w) = W_i \cap V(H)$.
\end{definition}

It will also be useful to have an operation that forms the ``reverse'' of a restriction; we define this below.

\begin{definition}\label{def:ext}
Let $G$ and $H$ be graphs such that $H \seq G$ and let $C_0$ be a partial contraction sequence of $H$. We say that a partial contraction sequence $C$ of $G$ is the \emph{extension} of $C_0$ to $G$ 
if $C[H] = C_0$ and no contraction in $C$ involves a vertex of $G - H$. 
When $G_i$ is the $i$-th trigraph in $C_0$, we denote by $G_i \uparrow G$ the $i$-th trigraph in $C$ (this makes sense since the lengths of $C_0$ and $C$ are the same).
\end{definition}

Finally, we introduce a notion that will be useful when dealing with reduction rules in the context of computing contraction sequences.

\begin{definition}\label{def:effective}
Let $G, G'$ be trigraphs. We say that the twin-width of $G'$ is \emph{effectively} at most the twin-width of $G$, denoted $\tww(G') \le_e \tww(G)$, if \textup{(1)} $\tww(G') \le \tww(G)$ and \textup{(2)} given a contraction sequence $C$ of $G$, a contraction sequence $C'$ of $G'$ of width at most $w(C)$ can be constructed in polynomial time.
If $\tww(G') \le_e \tww(G)$ and $\tww(G) \le_e \tww(G')$, then we say that the two graphs have \emph{effectively} the same twin-width, $\tww(G') =_e \tww(G)$.
\end{definition}

\smallskip
\noindent \textbf{Preliminary Observations and Remarks.}\quad
We begin by stating a simple brute-force algorithm for computing twin-width.

\begin{observation}\label{obs:brute-force}
An optimal contraction sequence of an $n$-vertex graph can be computed in time $2^{\bigoh(n\cdot \log n)}$.
\end{observation}
\begin{proof}
Each contraction sequence is defined by $n-1$ choices of a pair of vertices, and so the number of contraction sequences is $\bigoh((n^2)^n) = \bigoh(2^{2n\cdot \log n}) \le 2^{\bigoh(n\cdot \log n)}$. Moreover, computing the width of a contraction sequence can clearly be done in polynomial time.
\end{proof}

The following observation provides a useful insight into the optimal contraction sequences of trees.
\begin{observation}[{\cite[Section~3]{BonnetKTW22}}]\label{obs:contract-trees}
For any rooted tree $T$ with root $r$, there is a contraction sequence $C$ of $T$ of width at most $2$ such that the only contraction involving $r$ is the very last contraction in $C$.
\end{observation}

\section{The Square-Root Bound}
\label{sec:bound}

In this section, we prove that a graph with feedback edge number $k$ has twin-width at most $\ca O(\sqrt{k})$.
On a high level, the idea we will employ here builds on the preprocessing techniques originally introduced in the context of computing twin-width on tree-like graphs~\cite{BalabanGR24}: first we will contract the dangling trees, then the dangling paths, and for the final step we will use the following theorem of Ahn, Hendrey, Kim and Oum:

\begin{theorem}[\cite{general-bounds}]\label{thm:bound-edges}
If $G$ is a graph with $m$ edges, than the twin-width of $G$ is at most $\sqrt{3m} + o(\sqrt{m})$.
\end{theorem}

An issue we need to resolve before applying the aforementioned high-level approach is that Theorem~\ref{thm:bound-edges} only applies to graphs without red edges, whereas the trigraph $G$ we will obtain after dealing with the dangling trees and paths may contain these. The following lemma shows, using the properties of the red edges in $G$, that making all edges of $G$ black can only decrease the twin-width by a constant.

\begin{lemma}\label{lem:handling-red-edges}
Let $G$ be a trigraph with maximum red degree 2 such that each red edge in $G$ is incident to a vertex of degree at most 2.
If $G'$ is the graph obtained from $G$ by making all edges black, then $\tww(G) \le \tww(G') + 4$.
\end{lemma}
\begin{proof}
Let $C'$ be an optimal contraction sequence of $G'$, and let $C$ be the contraction sequence of $G$ obtained by following $C'$. We will prove that $w(C) \le w(C') + 4$.

Let $G_i$ be any trigraph in $C$ and let $G'_i$ be the trigraph in $C'$ such that $V(G_i) = V(G'_i)$.
Suppose for a contradiction that there are distinct vertices $u, v_1, v_2, v_3, v_4, v_5 \in V(G_i)$ such that for each $j \in [5]$, $uv_j$ is a red edge in $G_i$ but not in $G'_i$. Recall that $\beta(w)$ denotes the set of vertices contracted to $w$ (the bag of $w$), and observe that for each $j \in [5]$, there must be vertices $u_j, v_j^0 \in V(G)$ such that $u_j \in \beta(u)$ and $v_j^0 \in \beta(v_j)$, and $u_jv_j^0$ is an edge that is black in $G'$ but red in $G$.
Since $uv_j \notin R(G'_i)$, there must be either all edges or no edges between $\beta(u)$ and $\beta(v_j)$ in $G'$.
However, $u_jv_j^0 \in E(G')$, which means that for all $j, \ell \in [5]$, $u_jv_{\ell}^0$ is a black edge in $G'$ (and so it is an edge also in $G$).

Since all vertices of $G$ have red degree at most 2 and $u_jv_j^0 \in R(G)$ for each $j \in [5]$, there must be $a,b,c \in [5]$ such that $|\{u_a, u_b, u_c\}| = 3$. Now observe that each vertex in $\{u_a, u_b, u_c, v_a^0, v_b^0, v_c^0\}$ has degree at least 3 in $G$ (since $u_jv_{\ell}^0$ is an edge in $G$ for all $j, \ell \in \{a,b,c\}$). However, each red edge in $G$ has an endpoint of degree at most 2, which is a contradiction.

We have proven that the red degree of each vertex $u \in V(G_i) = V(G'_i)$ may be higher in $G_i$ than in $G'_i$ by at most 4, which proves that $w(C) \le w(C') + 4$.
\end{proof}

We are now ready to prove the square-root upper bound on twin-width.

\begin{theorem}\label{thm:upper}
There exists a function $f(k)\in \bigoh(\sqrt{k})$ such that every graph $G$ with feedback edge number $k$ has twin-width at most $f(k)$.
\end{theorem}

\begin{proof}
Let $F$ be a smallest feedback edge set of $G$ and assume $k := |F| > 0$.
We will prove the statement by constructing a contraction sequence for $G$ of width at most $f(k)\in \bigoh(\sqrt{k})$.
We begin by contracting each maximal dangling tree to a single vertex using Observation~\ref{obs:contract-trees}. After a maximal dangling tree has been contracted, we call the last remaining vertex a \emph{spike}. Whenever a vertex is adjacent to two spikes, we contract them together (the obtained vertex is still called a spike). Observe that throughout this process, no vertex has red degree higher than 2: this is ensured by Observation~\ref{obs:contract-trees} and the fact that a red neighbor of a vertex not in a dangling tree must be a spike.

Let $G^\alpha$ be the obtained trigraph and let $T$ be the tree obtained from $G^\alpha$ by removing all spikes and edges in $F$. Let $Q = \{u \in V(T) \sep u$ is incident to an edge of $F$ in $G^\alpha$ or $u$ has degree higher than $2$ in $T\}$.
Observe that $T - Q$ is a graph consisting of disjoint paths. Let $\ca P$ be the set of these paths. For each path $P = (u_1, \ldots, u_n)$ in $\ca P$, we do the following.

\begin{itemize}
\item If $n > 2$, then for each $i \in [2, n-1]$ such that $u_i$ has a spike $v$, contract $u_i$ and $v$ (do this in increasing order). If $u_1$ (resp. $u_n$) has a spike $v$, contract $v$ and $u_2$ (resp. $u_{n-1}$).
\item If $n > 3$, shorten $P$ to a path with exactly three vertices by repeatedly contracting neighboring vertices of $P - \{u_1, u_n\}$. 
\end{itemize}

\begin{figure}
\includegraphics[scale=1.4]{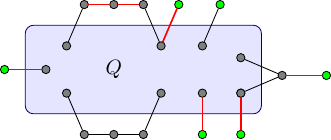}
\caption{A trigraph after processing the dangling trees and shortening the paths in $\ca P$. Spikes are colored in green. The edges between vertices of $Q$ are not depicted. Notice that one of the paths is black: this means it had no spikes and it has not been shortened. Also notice that one spike is attached by a black edge: it was a maximal dangling tree with only one vertex in $G$.}
\label{fig:square-root}
\end{figure}

Observe that throughout this process, no vertex has red degree higher than 2: a vertex in $Q$ has red degree at most 1 (its red neighbor must be a spike) and a vertex in a path $P$ either has a spike and at most one red neighbor in $P$ or at most two red neighbors in $P$. See Figure~\ref{fig:square-root} for an illustration.

Now we will count the number of edges in the obtained trigraph $G^\beta$.
Observe that $|Q|, |\ca P| \le 4k$.
First, observe that there are at most $5k$ edges in $G^\beta[Q]$: $k$ edges belonging to $F$ and at most $4k$ other edges since $G^\beta[Q] - F$ is a forest. In addition, each vertex of $Q$ may have a spike in $G^\beta$, which constitutes up to $4k$ other edges.
Second, let $P \in \ca P$. If the length of $P$ in $G$ is at least 2, then $P$ corresponds to at most four edges in $G^\beta$: at most two edges of the path itself and two edges connecting $P$ to the rest of the graph (i.e., to vertices of $Q$).
However, if $P$ is shorter in $G$, then its vertices may have spikes in $G^\beta$ and it may correspond to up to 5 edges: one edge of the path, two edges connecting it to $Q$, and two edges going to the spikes.
Hence, $\ca P$ adds at most $20k$ edges, and thus there are at most $29k$ edges in $G^\beta$.

Let $G^\gamma$ be the graph obtained from $G^\beta$ by changing the color of all edges to black.
By Theorem~\ref{thm:bound-edges}, $G^\gamma$ has twin-width at most $\sqrt{87k} + o(\sqrt{k})$.
Notice that $G^\beta$ satisfies the preconditions of Lemma~\ref{lem:handling-red-edges}, which means that $\tww(G^\beta) \le \tww(G^\gamma) + 4$.
Hence, the twin-width of $G^\beta$ is also at most $\sqrt{87k} + o(\sqrt{k})$, and the same also holds for the original graph $G$ (since the partial contraction sequence from $G$ to $G^\beta$ has width at most 2).
\end{proof}

We conclude the section by showing that Theorem~\ref{thm:upper} is asymptotically tight.

\begin{proposition}
\label{prop:lower}
There exists a function $f(k)\in \Omega(\sqrt{k})$ and an infinite class $\mathcal{G}$ of graphs such that for each $G\in \mathcal{G}$ with feedback edge number $k$, $\tww(G)\geq f(k)$.
\end{proposition}
\begin{proof}
Let $n$ be a prime power such that $n \equiv 1 \pmod 4$.
It is known that there exists an $n$-vertex ($(n-1)/2$)-regular graph $G$ (a so-called Paley graph) that has twin-width exactly $(n-1)/2$~\cite[Section 3]{general-bounds}.
Since $|E(G)| = (n^2-n)/4$ and the spanning forest of $G$ has at most $n-1$ edges, we know that the feedback edge number $k$ of $G$ is at least $(n^2-5n+4)/4 \in \Omega(n^2)$.

Let $\mathcal{G}$ be the class of all such $n$-vertex Paley graphs. For each $n$-vertex graph $G$ in $\mathcal{G}$, we have $k \in \Omega(n^2)$ and $\tww(G) \in \Theta(n)$. Let $f$ be the function which maps each $k$ to the minimum of $\{\tww(G)~|~G\in \mathcal{G}\text{ is a graph with feedback edge number }k\}$. Thus, for each $G\in \mathcal{G}$, $\tww(G)\geq f(k)$, and the aforementioned relationships between the number $n$ of vertices of that graph, $\tww(G)$ and $k$ guarantee that $f(k)\in \Omega(\sqrt{k})$, as desired. \end{proof}

\section{A Better Algorithm Parameterized by the Feedback Edge Number}
\label{sec:improvedfen}

We begin by recalling that the case of twin-width $2$ is known to already admit an exact nearly single-exponential fixed-parameter algorithm parameterized by the feedback edge number (see Theorem~\ref{thm:tww2} below), and thus here we focus our efforts on graphs with higher twin-width.

\begin{theorem}[\cite{BalabanGR24}]\label{thm:tww2}
If $G$ is a graph with feedback edge number $k$ and $\tww(G) \le 2$, then an optimal contraction sequence of $G$ can be computed in time $2^{\bigoh(k\cdot \log k)}+n^{\bigoh(1)}$.
\end{theorem}

Our algorithm uses the same initial preprocessing steps as the previous result of Balab\'an, Ganian and Rocton~\cite{BalabanGR24}. These are formalized through the following definition and theorem; note that in the approach we use here, we can use a slightly more general (and less technical) definition of \emph{tidy} $(H,\mathcal{P})$-\emph{graphs} than the preceding paper.

\begin{definition}\label{def:tidyHP}
A connected trigraph $G$ with $\tww(G) \ge 2$ is a tidy $(H,\ca P)$-graph if $\ca P$ is a non-empty set of dangling red paths in $G$, and there are two disjoint induced subtrigraphs of $G$, namely $H$ and $\sump\ca P$ (the disjoint union of all paths in $\ca P$), such that each vertex of $G$ belongs to one of them.
Moreover, if $u \in V(H)$ has a neighbor $v \in V(\sump\ca P)$ in $G$, then $u$ has black degree 0 in $G$, and $v$ is the only neighbor of $u$ in $\sump\ca P$.\end{definition}

The following theorem summarizes the results obtained in~\cite{BalabanGR24} that we will use in this section.

\begin{theorem}[\cite{BalabanGR24}, Theorem 17 + Corollary 20]\label{thm:preprocess}
There is a polynomial-time procedure which takes as input a graph $G$ with feedback edge number $k$ and either outputs an optimal contraction sequence of $G$ of width at most $2$, or a tidy $(H,\ca P)$-graph $G'$ with effectively the same twin-width as $G$ such that $|V(H)|\leq 112k$ and $|\ca P|\leq 4k$.
\end{theorem}

From here on, we pursue an entirely different approach than the one used to obtain the previous (non-elementary) kernel~\cite{BalabanGR24}. In Subsection~\ref{sub:givenCH}, we show how a tidy $(H, \ca P)$-graph can be contracted when the paths in $\ca P$ are long enough and a contraction sequence of $H$ is given. This is then used in Subsection~\ref{sub:sec4end}, where we describe a better algorithm for approximating twin-width parameterized by the feedback edge number (see Theorem~\ref{thm:better-algor}).

\subsection{Contracting an $(H,\ca P)$-Graph Using a Contraction Sequence for $H$}\label{sub:givenCH}

For this subsection, let us fix a tidy $(H,\ca P)$-graph $G$ and let $m := |\ca P|$. Assume that each $P \in \ca P$ satisfies $|V(P)| \ge 8m$ and let $F$ be the subtrigraph of $\sump\ca P$ induced by the vertices at distance at most $2m$ from $H$ in $G$.

Informally speaking, our goal now is to construct a ``good'' contraction sequence for such a trigraph $G$, see Corollary~\ref{cor:fen-improvement}. To achieve that, we need to describe some well-structured trigraphs obtained by a sequence of contractions from $G$, which we will call \emph{$G$-tidy trigraphs}, see the following Definition~\ref{def:gtidy}. 
An important property of a $G$-tidy trigraph is that all contractions happened either between two vertices of $H$ or two vertices of $F$ at the same distance from $H$ (see items~\ref{inv:basic} and~\ref{inv:levels}).

\begin{definition}\label{def:gtidy}
Let $G'$ be a trigraph obtained by a sequence of contractions from $G$ and let $H'$ (resp. $F'$) be the subtrigraph of $G'$ induced by the vertices $u$ such that $\beta(u)$ is a subset of $V(H)$ (resp. $V(F)$). 
We say that $G'$ is a \emph{$G$-tidy trigraph} if:
\begin{enumerate}
\item For $u \in V(G')$, we have $u \in V(H') \cup V(F')$ or $|\beta(u)| = 1$.\label{inv:basic}
\item Each $u \in V(H')$ has at most one neighbor outside of $H'$ in $G'$.\label{inv:outs}
\item For each $u \in V(F')$, all vertices in $\beta(u)$ have the same distance $d$ from $H$ in $G$. We say that $d$ is the \emph{level} of $u$.\label{inv:levels}
\item $F'$ is a forest such that all its vertices have degree at most 3 in $G'$. If $T$ is a connected component of $F'$, then:\label{inv:forest}
\begin{enumerate}
\item $T$ has exactly one vertex $r$ at level 1 (let us declare it the root of $T$).
\item The vertices of $T$ with degree $3$ in $G_i$ form a subtree $T'$ of $T$.\label{inv:subtree}
\item The vertices of $T'$ have level at most $|\beta(r)| - 1$, and either $T' = \emptyset$ or $r \in V(T')$.\label{inv:depth}
\end{enumerate}
\end{enumerate}
See Figure~\ref{fig:gtidy} for an illustration.
\end{definition}

\begin{figure}[ht]
\scalebox{1.2}{
\begin{tikzpicture}
\clip(-2,-2.4) rectangle (2.5,1.8);.

\tikzset{
    ff/.style = {draw,black, circle, fill=blue!40!white, minimum width=4pt, inner sep=0pt}}
    
\tikzset{
    deg3/.style = {draw,black, circle, fill=blue, minimum width=4pt, inner sep=0pt}}
    
\tikzset{
    deep/.style = {draw,black, circle, fill=green, minimum width=4pt, inner sep=0pt}}    

\def\points{14};    
    
\tikzset{st1/.style={postaction={decorate,
decoration={markings, mark= at position 1/\points with {\node [ff]{};}, 
mark= at position 2/\points with {\node [ff]{};},
mark= at position 3/\points with {\node [ff]{};}, 
mark= at position 4/\points with {\node [ff]{};}, 
mark= at position 5/\points with {\node [ff]{};}, 
mark= at position 6/\points with {\node [deep]{};},
mark= at position 7/\points with {\node [deep]{};},
mark= at position 8/\points with {\node [ff]{};},
mark= at position 9/\points with {\node [ff]{};}, 
mark= at position 10/\points with {\node [ff]{};},
mark= at position 11/\points with {\node [ff]{};}, 
mark= at position 12/\points with {\node [ff]{};},
mark= at position 13/\points with {\node [ff]{};},
}}}}
      
\def\ppp{12};    
    
\tikzset{st2/.style={postaction={decorate,
decoration={markings, mark= at position 1/\ppp with {\node [ff]{};}, 
mark= at position 2/\ppp with {\node [ff]{};},
mark= at position 3/\ppp with {\node [ff]{};}, 
mark= at position 4/\ppp with {\node [ff]{};}, 
mark= at position 5/\ppp with {\node [ff]{};}, 
mark= at position 6/\ppp with {\node [deep]{};},
mark= at position 7/\ppp with {\node [deep]{};},
mark= at position 8/\ppp with {\node [ff]{};},
mark= at position 9/\ppp with {\node [ff]{};}, 
mark= at position 10/\ppp with {\node [ff]{};},
mark= at position 11/\ppp with {\node [ff]{};}, 
}}}}

\tikzset{st3/.style={postaction={decorate,
decoration={markings, mark= at position 1/11 with {\node [ff]{};}, 
mark= at position 2/11 with {\node [ff]{};},
mark= at position 3/11 with {\node [ff]{};}, 
mark= at position 4/11 with {\node [ff]{};}, 
mark= at position 5/11 with {\node [deep]{};}, 
mark= at position 6/11 with {\node [deep]{};},
mark= at position 7/11 with {\node [ff]{};},
mark= at position 8/11 with {\node [ff]{};},
mark= at position 9/11 with {\node [ff]{};}, 
mark= at position 10/11 with {\node [ff]{};},
}}}}

    \definecolor{Hcolor}{RGB}{182,255,193}
         \fill[gray!20!white] (0,0) circle (1 cm);
    \draw (0,0) circle (1 cm);
         \def\n{3}
    \node at (0,0) {\large $H'$};
         \begin{scope}[every node/.style={draw, circle, fill=gray, minimum width=4pt, inner sep=0pt}]
   
    \foreach \i in {1,...,\n} {
                 \node (a\i) at (360/\n*\i: 1 cm) {};}
    
    \foreach \i in {2,...,\n} {
                 \node[deg3] (b\i) at (360/\n*\i: 1*1.4 cm) {};\draw[red] (a\i)-- (b\i);}
    \end{scope}
   
    \node[ff] (d0) at (b3) [above = 10pt]{};
    \node[deg3] (d1) at (b3) [below = 10pt]{};
    \draw[red] (d0)-- (b3)--(d1);
    
    \draw[red, st1] (b2) to[out=180,in=90, distance=70pt] (a1);
    \draw[red, st2] (b2) to[out=290,in=270, distance=50pt] (d1);
    \draw[red, st3] (d1) to[out=315,in=90, distance=70pt] (d0);
\end{tikzpicture}
}
\caption{An illustration of Definition~\ref{def:gtidy} when $m = 3$. The depicted $G$-tidy trigraph $G'$ consists of $H'$: vertices colored in grey, $F'$: vertices colored in blue (degree-3 vertices in darker shade), and the remaining vertices are colored in green. The edges inside of $H'$ are not depicted (there can be both red and black edges). Note that instead of each pair of green vertices, there should be at least 12 of them (because each path in $\ca P$ should contain at least $8m = 24$ vertices).
\label{fig:gtidy}}
\end{figure}
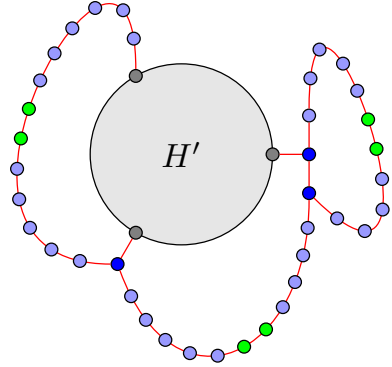

Now we show how $G$ can be reduced to a $G$-tidy trigraph with $H'$ being a single-vertex graph; this will be the first part of the proof of Corollary~\ref{cor:fen-improvement}. Note that the assumption that $C_H$ is given will be later handled in the proof of Lemma~\ref{lem:fen-impr-final-lem}.

\begin{lemma}\label{lem:contrH}
Given a contraction sequence $C_H$ of $H$, one can compute a partial contraction sequence of width $\max(w(C_H)+1, 4)$ from $G$ to a $G$-tidy trigraph $G'$ with $|V(H')| = 1$, in polynomial time.
\end{lemma}
\begin{proof}
For each $i \in [|V(H)|]$, we will construct a partial contraction sequence $C_i$ from $G$ to a $G$-tidy trigraph $G_i$ such that $w(C_i) \le \max(w(C_H)+1, 4)$ and the restriction of $C_i$ to $H$ will be the prefix of $C_H$ of length $i$.
We will denote the subtrigraphs of $G_i$ corresponding to $H'$ and $F'$ (see Definition~\ref{def:gtidy}) $H_i$ and $F_i$, respectively.
We define $C_1 = (G)$, i.e., $C_1$ is the trivial partial contraction sequence with no contractions. It can be easily verified that $G_1 := G$ is a $G$-tidy trigraph. In particular, the forest $F_1 := F$ consists of $2m$ disjoint paths.

Suppose that we have constructed $C_i$ for some $i < |V(H)|$. Let $u, v \in V(H_i)$ be the two vertices contracted in $H_{i+1}$ (which is the successor of $H_i$ in $C_H$). If $u$ or $v$ does not have a neighbor outside of $H_i$ in $G_i$, then we define $G_{i+1}$ to be the trigraph obtained from $G_i$ by contracting $u$ and $v$. Clearly, $G_{i+1}$ is a $G$-tidy trigraph and $C_{i+1}$ (the sequence obtained by prolonging $C_i$ with $G_{i+1}$) has the required properties.
Now suppose that both $u$ and $v$ have a neighbor outside of $H_i$ in $G_i$. In this case, we cannot simply contract them because the new vertex would have two neighbors outside of $H_{i+1}$, violating condition~\ref{inv:outs} of Definition~\ref{def:gtidy}.

Let $T_u$ and $T_v$ be the two connected components of $F_i$ with roots adjacent to $u$ and $v$, respectively. Informally, we need to merge $T_u$ and $T_v$ before we can contract $u$ and $v$. Let $T \in \{T_u, T_v\}$ be a tree with root $r$. If $T$ contains no degree-3 vertices, we do nothing (we always mean degree in $G_i$). Otherwise, let $w \in V(T)$ be the deepest degree-3 vertex such that all its ancestors in $T$ have degree 3.
By item~\ref{inv:levels} of Definition~\ref{def:gtidy}, $|\beta(r)| \le 2m$ because $F$ contains exactly $2m$ vertices at level 1 (2 for each path in $\ca P$). 
Hence, by item~\ref{inv:depth}, the level of $w$ is less than $2m$, and so $w$ has two children $x$ and $y$ in $T$, both of degree 2. We contract $x$ and $y$ (note that the obtained vertex $xy$ has red degree 3, and the red degree of $w$ drops to 2). We repeat this process as long as such vertex $w$ exists (crucially, $xy$ cannot be chosen as the next $w$ because its parent has degree 2). Afterwards, we contract the roots $r_u$, $r_v$ of $T_u$ and $T_v$, and finally, we contract $u$ and $v$ into $uv$.

Let $C_{i+1}$ be the partial contraction sequence of $G$ obtained by prolonging $C_i$ with the contractions described in the previous paragraph.  Let us show that $w(C_{i+1}) \le \max(w(C_H)+1, 4)$. By the assumption about $C_i$, it suffices to discuss red degrees in each trigraph $G'$ between $G_i$ and $G_{i+1}$ (which is the last trigraph in $C_{i+1}$). Clearly, any descendant of $H$ in $G'$ has red degree at most $w(C_H)+1$ (it is crucial that $u$ and $v$ are contracted after $r_u$ and $r_v$). Any other vertex of $G'$ has red degree at most 3, except for the vertex obtained by contracting $r_u$ and $r_v$, whose red degree is 4 (but it drops to 3 when $u$ and $v$ are contracted).

Finally, we need to show that $G_{i+1}$ is $G$-tidy. It is easy to see that $G_{i+1}$ satisfies the first three items of Definition~\ref{def:gtidy}. To prove that $G_{i+1}$ satisfies item~\ref{inv:forest}, observe that $F_{i+1}$ is indeed a forest: it contains the same trees as $F_i$, except that $T_u$ and $T_v$ have been merged into a new tree $T$ with root $r$.
More precisely, $T$ is isomorphic to the tree obtained from the disjoint union of $T_u$ and $T_v$ by first adding a new vertex $r$ and edges $rr_u, rr_v$, and second removing all leaves. Since $r$ has degree 3 in $G_{i+1}$, the highest level of a degree-3 vertex in $T$ is one higher than in the union of $T_u$ and $T_v$ in $G_i$.
Since $|\beta(r)| = |\beta(r_u)| + |\beta(r_v)|$ and both of these summands are at least 1, we get that $G_{i+1}$ satisfies item~\ref{inv:depth}, which concludes the proof.
\end{proof}

To prove Corollary~\ref{cor:fen-improvement}, we now show that the $G$-tidy trigraph given by Lemma~\ref{lem:contrH} can be contracted to a single vertex (without creating vertices with high red degree). Note that the following proof is inspired by the proof of Theorem~7 in~\cite{BergeBD22}.

\begin{lemma}\label{lem:contrP}
If $G'$ is a $G$-tidy trigraph $G'$ with $|V(H')| = 1$, then a contraction sequence of $G'$ of width at most 4 can be computed in polynomial time.
\end{lemma}
\begin{proof}

First, observe that by Definition~\ref{def:tidyHP}, all edges in $G'$ are red. By item~\ref{inv:outs} of Definition~\ref{def:gtidy}, the only vertex $u$ of $H'$ has a single neighbor $r$.
By Definition~\ref{def:tidyHP}, $G$ is connected; hence, also $G'$ is connected. This implies that $F'$ is a tree. Let $T$ be the subtree of $F'$ induced by the vertices with degree 3 in $G'$.
Let us begin by contracting $u$ and $r$, obtaining a trigraph $G^* := G' - u$. 

Observe that the depth of $T$ is at most $2m - 1$ because $r$ contains $2m$ vertices in its bag (by property~\ref{inv:depth}).
Consider a path $P \in \ca P$ and observe that the descendants of at most $4m$ vertices of $P$ belong to $T$ in $G^*$ ($2m$ from each side). 
Hence, $G^* - T$ consists of disjoint dangling red paths, each with at least $4m$ vertices (since each $P \in \ca P$ satisfies $|V(P)| \ge 8m$).
Let $\ca P'$ be the set of these red paths in $G^*$.

Let $P \in \ca P'$ and let $u$ and $u'$ be the endpoints of $P$. Let $v, v' \in V(T)$ be the neighbors of $u$ and $u'$ in $T$, respectively, and let $Q$ be the path connecting $v$ and $v'$ in $T$. Since the depth of $T$ is at most $2m -1$, we know that $Q$ contains at most $4m-1$ vertices. Let us shorten $P$ so that it has the same length as $Q$ (by repeatedly contracting consecutive vertices). Let $(u_1 = u, \ldots, u_p = u')$ and $(v_1 = v, \ldots, v_p = v')$ be the sequences of vertices of $P$ and $Q$ in the natural orders. Now for each $i \in [p]$ in increasing order, contract $u_i$ and $v_i$, and observe that the obtained trigraph is isomorphic to $G^* - P$. Repeat this for all paths $P \in \ca P'$, obtaining a trigraph isomorphic to $T$
, which has twin-width at most 3 and can be contracted as per Observation~\ref{obs:contract-trees}.
Finally, observe that during a contraction of a path in $P \in \ca P'$, there is never a vertex with red degree higher than $4$. Indeed, after contracting $u_i$ and $v_i$ for $i \in [p-1]$, the obtained vertex has at most four red neighbors: at most three in $T$ plus $u_{i+1}$.
\end{proof}

\subsection{Wrapping up the Proof}\label{sub:sec4end}

In the previous subsection, we proved Lemmas~\ref{lem:contrH} and~\ref{lem:contrP}, which together imply the following corollary.

\begin{corollary}\label{cor:fen-improvement}
Let $G$ be a tidy $(H,\ca P)$-graph such that each $P \in \ca P$ satisfies $|V(P)| \ge 8\cdot|\ca P|$. Given a contraction sequence $C_H$ of $H$, one can compute a contraction sequence of $G$ of width $\max(w(C_H)+1, 4)$, in polynomial time.
\end{corollary}

Now we are able to show that if we shorten all long paths in a tidy $(H, \ca P)$-graph, then the twin-width increases by at most 1.  

\begin{lemma}\label{lem:fen-impr-final-lem}
Let $G$ be a tidy $(H_0, \ca P_0)$-graph such that $\tww(G) \ge 3$, let $m = |\ca P_0|$ and let $G'$ be the trigraph obtained from $G$ by shortening each path $P\in\ca P_0$ with more than $8\cdot m$ vertices to length exactly $8\cdot m-1$. Then $\tww(G') \le \tww(G) + 1$.
\end{lemma}
\begin{proof}
We begin by handling short paths in $\ca P_0$: let $\ca P_{short} = \{P \in \ca P_0\colon |V(P)| < 8m\}$, let $H$ be the union of $H_0$ and $\sqcup \ca P_{short}$ (including the edges between them), and let $\ca P = \ca P_0 \setminus \ca P_{short}$. Clearly, $G$ is also a tidy $(H, \ca P)$-graph. Also observe that $G'$ is a tidy $(H, \ca P')$-graph (where $\ca P'$ is the set of paths obtained from $\ca P$ by shortening each path in it).

We want to construct a contraction sequence $C'$ of $G'$ of width at most $\tww(G) + 1$ from an optimal contraction sequence $C$ of $G$. Let $C_H$ be the restriction of $C$ to $H$; clearly, $w(C_H) \le \tww(G)$. Since $\tww(G) \ge 3$, it suffices to apply Corollary~\ref{cor:fen-improvement} on $G'$ using $C_H$, which yields the desired contraction sequence $C'$.
\end{proof}

Finally, we are able to prove the main result of this section.

\begin{theorem}\label{thm:better-algor}
Given a graph $G$ with feedback edge number $k$, a trigraph $G'$ of size $\ca O(k^2)$ such that
$\tww(G) \le \tww(G') \le \tww(G)+1$
can be computed in polynomial time. Moreover, a contraction sequence for $G$ of width at most $\tww(G)+1$ can be computed in time $2^{\bigoh(k^2\cdot \log k)}+n^{\bigoh(1)}$.
\end{theorem}

\begin{proof}
First, we use Theorem~\ref{thm:tww2} to check whether $\tww(G) \le 2$ (if yes, $G'$ can be any constant-size graph with the same twin-width as $G$). From now on, assume $\tww(G) \ge 3$. Now let us use Theorem~\ref{thm:preprocess} to obtain a tidy $(H, \ca P)$-graph $G_1$ with effectively the same twin-width as $G$ such that $|V(H)| \le 112k$ and $|\ca P| \le 4k$. Let $G'$ be the trigraph obtained when Lemma~\ref{lem:fen-impr-final-lem} is applied on $G_1$. By Lemma~\ref{lem:fen-impr-final-lem}, $\tww(G') \le \tww(G_1) + 1$. Conversely, $\tww(G') \ge \tww(G_1)$ because there is a partial contraction sequence $C_1$ from $G_1$ to $G'$ of width at most $\tww(G')$; it suffices to shorten paths of $\ca P$ that are shorter in $G'$ than in $G_1$ by contracting consecutive vertices. Hence, we indeed have $\tww(G) \le \tww(G') \le \tww(G)+1$.

Next, let us examine the size of $G'$. By Lemma~\ref{lem:fen-impr-final-lem}, each of the $4k$ paths in $\ca P$ has at most $8\cdot 4k$ vertices in $G'$. Hence, we obtain $|V(G')| \le 128k^2 + 112k \in \ca O(k^2)$ as required.

Finally, let us show how a contraction sequence for $G$ of width at most $\tww(G)+1$ can be computed.
If $\tww(G) \le 2$, then this contraction sequence is provided by Theorem~\ref{thm:tww2}. Otherwise, observe that an optimal contraction sequence $C'$ of $G'$ can be computed in time $2^{\bigoh(k^2\cdot \log k)}$ by Observation~\ref{obs:brute-force}. Next we concatenate $C'$ and $C_1$ (which is defined above and can be computed trivially) to obtain a contraction sequence of $G_1$ of width at most $\tww(G_1)+1$. We conclude using the effectiveness part of $\tww(G) =_e \tww(G_1)$ (see Definition~\ref{def:effective}).
\end{proof}

\section{A Fixed-Parameter Algorithm Parameterized by Vertex Integrity}\label{sec:vi}

\newcommand{\con}{\ca C}
In this section, we design an FPT 2-approximation algorithm for computing twin-width when parameterized by the vertex integrity, see Theorem~\ref{thm:param-by-VI}. 

\subsection{Initial Setup and Overview}\label{sub:5setup}

For the following, it will be useful to recall the definition of vertex integrity and its associated decomposition presented in Section~\ref{sec:prelims}. Let us fix a graph $G$ and a choice of $S\subseteq V(G)$ witnessing that the vertex integrity of an input graph $G$ is $p$, and let $\con$ be the set of connected components of $G - S$. We assume without loss of generality that $G$ is connected, as the twin-width of a graph is the maximum twin-width of its connected components. 
We now define a notion of ``component-types'' which intuitively captures the equivalence between components which exhibit the same outside connections and internal structure.

\begin{definition}\label{def:equivalence}
We say that two graphs $H_0, H_1 \in \con$ are \emph{twin-blocks}, denoted $H_0 \sim H_1$, if there exist a \emph{canonical} isomorphism $\alpha$ from $H_0$ to $H_1$ such that for each vertex $u\in V(H_0)$ and each $v\in S$, $uv\in E(G)$ if and only if $\alpha(u)v \in E(G)$. 
Clearly, $\sim$ is an equivalence relation.
\end{definition}

In a nutshell, our algorithm first computes an optimal contraction sequence $C'$ for a subgraph $G'$ of $G$ that is obtained by keeping only a bounded number of twin-blocks from each equivalence class, and then uses $C'$ to obtain a contraction sequence for $G$ of width at most $2\cdot \tww(G') \le 2\cdot \tww(G)$. In the following definition, we introduce terminology related to subgraphs of $G$.

\begin{definition}\label{def:G'}
Let $G'$ be an induced subgraph of $G$.
\begin{itemize}
\item We say that $G'$ is \emph{$\con$-respecting} if $S \seq V(G')$ and for each $H \in \con$, either $H \seq G'$ or $V(H) \cap V(G') = \emptyset$.
\item We say that an equivalence class $[H_0]$ of $\sim$ is \emph{large in $G'$} if $|\ca H| \ge f(p)$, where $\ca H = \{H \in [H_0] \sep H \seq G'\}$ and $f(p) = 2^{7p^3}$.
\item We say that $G'$ is the \emph{reduced graph} of $G$ if it is obtained from $G$ by removing all but $f(p)$ twin-blocks from each large class of $\sim$.
\end{itemize}
\end{definition}

Let us now bound the size of the reduced graph $G'$.

\begin{observation}\label{obs:G'-size}
If $G'$ is the reduced graph of $G$, then $|V(G')| \le p+p^2\cdot f(p)\cdot 2^{2p^2} \in 2^{\ca O(p^3)}$.
\end{observation}
\begin{proof}
First, let us compute the size of $\con/{\sim}$. Each $H \in \con$ has at most $p$ vertices, which means that the number of non-isomorphic graphs in $\con$ can be upper-bounded by $p \cdot 2^{p^2}$. Since $|S| \le p$, there are at most  $p^2$ possible edges between $S$ and each $H \in \con$. Hence, $|\con/{\sim}| \le p \cdot 2^{2p^2}$.
Because $|V(H)| \le p$ for each $H \in \con$ and by definition of $G'$, the union of each class of $\sim$ contains at most $p \cdot f(p)$ vertices. Finally, we again use that $|S| \le p$.
\end{proof}

The core of our algorithm is the following lemma, which we will prove in Subsection~\ref{sub:sec5}:

\begin{lemma}\label{lem:CfromC'}
If $G'$ is the reduced graph of $G$, then given a contraction sequence $C'$ for $G'$ of width $t$, we can compute a contraction sequence for $G$ of width at most $2t$ in polynomial time.
\end{lemma}

Let us now show how we can use this lemma to design the desired algorithm:

\begin{theorem}\label{thm:param-by-VI}
If $G$ is a graph with vertex integrity $p$, then a contraction sequence for $G$ of width at most $2\cdot \tww(G)$ can be computed in time $g(p) \cdot n^{\ca O(1)}$, where $g$ is an elementary function.
\end{theorem}

\begin{proof}
The first step of the algorithm is to compute an optimal vertex-integrity decomposition of $G$. As noted already in Section~\ref{sec:prelims}, this can be done in time $\bigoh(p^{p+1}\cdot n)$~\cite{DrangeDH16}. Using this decomposition, we can compute the reduced graph $G'$ of $G$ in linear time. Next, we can compute an optimal contraction sequence $C'$ of $G'$, using Observation~\ref{obs:brute-force}. Since the size of $G'$ is bounded (see Observation~\ref{obs:G'-size}), we deduce that computing $C'$ takes time $g(p) \in \exp(\exp(\ca O(p^3)))$, where $\exp(x) = 2^x$. 

Finally, we apply Lemma~\ref{lem:CfromC'} to compute in polynomial time a contraction sequence $C$ for $G$ of width at most $2\cdot w(C') = 2\cdot \tww(G')$. Since $G'$ is an induced subgraph of $G$, we know $\tww(G') \le \tww(G)$, which implies the desired bound $w(C) \le 2\cdot \tww(G)$.
\end{proof}

\subsection{Extending a contraction sequence from $G'$ to $G$}\label{sub:sec5}

This subsection is dedicated to proving Lemma~\ref{lem:CfromC'}. Recall that we have fixed a graph $G$ and a set $S \seq V(G)$, and that $\con$ is the set of connected components of $G-S$.
Let us begin with several technical definitions.

\begin{definition}\label{def:merged}
Let $G'$ be a $\con$-respecting graph, let $H_0$ and $H_1$ be distinct twin-blocks (with canonical isomorphism $\alpha$) such that $H_0, H_1 \seq G'$, and let $G^*$ be any trigraph obtained from $G'$ by a sequence of contractions. We say that $H_0$ and $H_1$ are \emph{merged} in $G^*$ if, for each $u\in V(H_0)$, there is a vertex $v \in V(G^*)$ such that $u, \alpha(u) \in \beta(v)$.
\end{definition}

It might be confusing that in the following definition, we consider a $\con$-respecting graph and a graph $H \in \con$ that is \emph{not} its induced subgraph. The reason for this is that later we will show that, under some conditions, $H$ can be ``added'' without increasing the twin-width too much. In fact, all such graphs $H$ will be progressively added until all of them are present (and the obtained graph is the whole $G$). To formalize the process of adding $H$, we will use Definition~\ref{def:ext} to create an extension of a contraction sequence to a sequence with $H$ ``appended'' to all trigraphs.

\begin{definition}\label{def:critical}
Let $G'$ be a $\con$-respecting graph, let $H \in \con$ be such that $H \nsubseteq G'$, let $C' = (G_1', G_2', \ldots)$ be a contraction sequence of $G'$. 
\begin{itemize}
\item We say that a trigraph $G_i'$ in $C'$ is the \emph{$C'$-critical trigraph} for $H$ if $i$ is the least index such that some vertex of $H$ has a red neighbor in $G_i' \uparrow G$.
\item If $G_i'$ is the $C'$-critical trigraph for $H$, then we say that a trigraph $G_j'$ is \emph{$C'$-safe} for $H$ if $j < i$ and there are two graphs $H', H'' \in [H]_\sim$ that are merged in $G_j'$.
\end{itemize}
\end{definition}

We will show that for each $H$ and $C'$ (as in Definition~\ref{def:critical}), there is a $C'$-safe trigraph for $H$. The first step towards this is to show that if $H$ has many twin-blocks in $G'$, then there are two twin-blocks of $H$ merged in the $C'$-critical trigraph $G^*$ for $H$.
Intuitively, if the twin-blocks of $H$ were not ``sufficiently-merged'' in $G^*$, then some vertex of $S$ would have high red degree because the existence of a red edge between $S$ and $H$ (see the definition of $C'$-critical) implies red edges between $S$ and all twin-blocks of $H$.

\begin{lemma}\label{lem:merged_copies}
If $G'$ is a $\con$-respecting graph, $C'$ is a contraction sequence of $G'$, $H \in \con$ is a graph such that $H \nsubseteq G'$, the class $\ca H := [H]_\sim$ is large in $G'$, and $G^*$ is the $C'$-critical trigraph for $H$, then there are two graphs $H', H'' \in [H]_\sim$ that are merged in $G_i'$.
\end{lemma}

\begin{proof}
Let $I = [f(p)]$ and let $H_1, \ldots, H_{f(p)} \in \ca H$ be distinct graphs such that $H_i \seq G'$ for each $i \in I$ (using the fact that $\ca H$ is large in $G'$). For $i \in I$ and $u \in V(H)$, let $u_i := \alpha(u)$, where $\alpha \colon V(H) \rightarrow V(H_i)$ is a canonical isomorphism. 
Let $u, v \in V(G^* \uparrow G)$ be two vertices such that $u \in V(H)$ and $uv$ is a red edge in $G^* \uparrow G$. By Definition~\ref{def:critical}, such vertices $u$ and $v$ exist, and by definition of vertex integrity, $v$ is a descendant of $S$. Let $d := 2^{p+1} + 1$.
We shall prove by induction that the following claim holds. 

\begin{claim}\label{cl:merge}
For each $a \in [0, p-1]$, there is a set $I_a \seq I$ of size at least $f(p) / d^{pa+1}$ such that for each $i,j \in I_a$ and each vertex $w \in V(H)$ at distance at most $a$ from $u$ in $H$, there is a vertex $x \in V(G^*)$ such that $w_i, w_j \in \beta(x)$.
\end{claim}

Observe that this statement implies that $H_i$ and $H_j$ for any $i,j \in I_{p-1}$ are merged in $G^*$ because the diameter of $H$ is at most $p - 1$.

\begin{claimproof}[Proof of Claim~\ref{cl:merge}]
Let us start by proving Claim~\ref{cl:merge} for $a = 0$.
Let $U = \{u_i \sep i \in I\}$ and observe that for each $i \in I$, the descendant $u_i'$ of $u_i$ is a red neighbor of $v$ in $G^*$, by Definition~\ref{def:equivalence} (unless $u_i' = v$). However, the red degree of $v$ in $G^*$ is at most $\tww(G') \le 2^{p+1}$ (because the treewidth of $G'$ is at most the vertex integrity of $G'$, and the twin-width is bounded by treewidth, see \cite{DBLP:conf/wg/JacobP22}). Hence, the vertices of $U$ are present in the bags of at most $d$ vertices in $G^*$ (note that some vertices of $U$ may be in the bag of $v$), which means that there is a vertex $w \in V(G^*)$ with at least $f(p) / d$ vertices of $U$ in its bag.
Now it suffices to set $I_0 := \{i \in I \sep u_i \in \beta(w)\}$.
This concludes the proof of the base case of the induction.

For the induction step, suppose that Claim~\ref{cl:merge} holds for some $a \in [0, p-2]$, i.e., there is a set $I_a \seq I$ with the described properties.
Let $D_a, D_{a+1} \seq V(H)$ be the sets of vertices at distance exactly $a$ or $a+1$ from $u$ in $H$, respectively.
Let $w \in D_{a+1}$ and $x \in D_a$ be two neighbors in $H$.
Let $x'$ be the descendant of $x_i$ in $G^*$ for some $i \in I_a$ (or, equivalently, for each $i \in I_a$, by the induction hypothesis), let $W = \{w_i \sep i \in I_a\}$, and let $w_i'$ be the descendant of $w_i$ in $G^*$ (for any $i \in I_a$). Observe that $x'w_i'$ is a red edge of $G^+$, unless $x' = w_i'$. Using the same argument as in the base case, $x'$ has red degree at most $d-1$ in $G^*$, which means that the vertices of $W$ are present in the bags of at most $d$ vertices in $G^*$.

Since $|D_{a+1}| \le p$, $I_a$ can be partitioned into at most $d^p$ parts such that if $i, j \in I_a$ are in the same part, then for each vertex $w \in D_{a+1}$, $w_i$ and $w_j$ are in the bag of the same vertex in $G^*$. Hence, one of these parts has size at least $|I_a| / d^p$, and we choose it to be $I_{a+1}$. A simple computation shows that $I_{a+1}$ satisfies Claim~\ref{cl:merge}.
\end{claimproof}

Finally, we only need to verify that $|I_{p-1}| \ge f(p) / d^{p(p-1) +1} \ge 2$. Recall that $f(p) = 2^{7p^3}$ and $d = 2^{p+1} +1 \le 2^{3p}$ since $p \ge 1$. Since $p(p-1) + 1 \le 2p^2$, we get $|I_{p-1}| \ge 2^{7p^3}/2^{6p^3} \ge 2$, which concludes the proof.
\end{proof}

Now we need to take a closer look at $S$.

\begin{definition}\label{def:H-equiv}
Let $H \in \con$ and $u, v \in S$. We say that $u$ and $v$ are \emph{$H$-equivalent}, denoted $u \sim_H v$, if and only if for each $w \in V(H)$, $uw \in E(G) \Leftrightarrow vw \in E(G)$. Let $S^H \seq S$ be the set of vertices with at least one neighbor in $H$ (in $G$). If $G_i'$ is a trigraph in a contraction sequence of a $\con$-respecting graph, then we denote by $S_i^H$ the set of descendants of $S^H$ in $G_i'$.
\end{definition}

A crucial observation is that before the $C'$-critical trigraph for $H$, only very restricted contractions may involve vertices of $S^H$ (so that a red edge to $H$ does not appear).

\begin{observation}\label{obs:subset}
If $G'$ is a $\con$-respecting graph, $C'= (G_1', G_2', \ldots)$ is a contraction sequence of $G'$, $H \in \con$ is a graph such that $H \nsubseteq G'$, $G_i'$ is the $C'$-critical trigraph for $H$, and $j < i$, then for each $u \in S_j^H$, the bag $\beta(u)$ is a subset of an equivalence class of $\sim_H$.
\end{observation}

\begin{proof}
Suppose for contradiction that there is $u \in S_j^H$ such that $\beta(u)$ is not a subset of an equivalence class of $\sim_H$. If $\beta(u) \nsubseteq S$, then clearly all neighbors of $u$ in $H$ (in $G_j' \uparrow G$) would be red, a contradiction with $j < i$ and the choice of $i$. Hence, assume $\beta(u) \seq S$. If there are $v_0, v_1 \in \beta(u)$ such that $v_0 \nsim_H v_1$, then there is a vertex $w \in H$ that has exactly one neighbor in $\{v_0, v_1\}$ in $G$, by Definition~\ref{def:H-equiv}. Thus, $uw$ is a red edge in $G_j' \uparrow G$, again a contradiction
\end{proof}

Using Observation~\ref{obs:subset}, we can prove the existence of a $C'$-safe trigraph.

\begin{lemma}\label{lem:safe-exists}
If $G'$, $C'$, $H$ and $G_i'$ are as in Observation~\ref{obs:subset} and the equivalence class $[H]_\sim$ is large in $G'$, then $G_{i-1}'$ is a $C'$-safe trigraph for $H$.
\end{lemma}

\begin{proof}
By Definition~\ref{def:critical}, it suffices to show that there are two graphs $H', H'' \in [H]_\sim$ that are merged in $G_{i-1}'$.
By Lemma~\ref{lem:merged_copies}, we know that such merged graphs $H'$ and $H''$ exist for $G'_i$.
Let $u,v \in V(G_{i-1}')$ be the two vertices that are contracted in $G_i'$, and suppose for contradiction that $H'$ and $H''$ are not merged already in $G_{i-1}'$.
This implies that $u$ and $v$ are both descendants of $H' \cup H''$.
However, by Observation~\ref{obs:subset}, $u, v \notin S_{i-1}^H$.
This is a contradiction with Definition~\ref{def:critical} because the contraction creating $G_i'$ must involve a vertex of $S^H$ so that a red edge incident to $H$ can appear in $G_i' \uparrow G$.
\end{proof}

Now we are ready to show how a contraction sequence $C'$ of $G'$ can be modified when a graph $H \in \con $ is added to $G'$. Unfortunately, we cannot do that without increasing the width. Since our goal is to eventually add many graphs $H \in \con$, we need to keep the increase under control, for which we use the following definition.

\begin{definition}\label{def:progress}
A contraction sequence $C = (G_1,\ldots, G_n)$ has \emph{progressive width} $(a \rightarrow_i b)$ if the width of $(G_1, \ldots, G_{i-1})$ is at most $a$ and the width of $(G_{i}, \ldots, G_n)$ is at most $b$.
\end{definition}

\begin{lemma}\label{lem:one_new_H}
Let $G'$ be a $\con$-respecting graph, let $C' = (G_1', G_2',\ldots)$ be a contraction sequence of $G'$, let $H \in \con$, $H \nsubseteq G'$ be such that $\ca H := [H]_\sim$ is large in $G'$, let $G^+ = G[V(G') \cup V(H)]$, and let $G'_{i-1}$ be a $C'$-safe trigraph for $H$. If $C'$ has progressive width $(t \rightarrow_{i} 2t)$, then we can construct in polynomial time a contraction sequence $C^+$ for $G^+$ of progressive width $(t \rightarrow_i 2t)$.
Moreover, if $j < i$, then $G_j' \uparrow G^+$ is the $j$-th trigraph in $C^+$.
\end{lemma}

\begin{proof}
By Definition~\ref{def:critical}, there are $H', H'' \in \ca H$ (such that $H', H'' \seq G'$) that are merged in $G_{i-1}'$.
Let $\iota\colon H \rightarrow H'$ be a canonical isomorphism, let $C'_{<i}$ be the prefix of $C'$ of length $i-1$, and let $C_H$ be the partial contraction sequence of $H$ isomorphic to $C'_{<i}[H']$ with an isomorphism induced by $\iota$\footnote{Formally, an isomorphism from $(G = G_1,\ldots,G_n)$ to $(H = H_1, \ldots, H_n)$ induced by an isomorphism $\alpha\colon G \rightarrow H$ is a sequence of isomorphisms $\alpha_i\colon G_i \rightarrow H_i$ such that for each $i \in [n]$ and $u \in V(G_i)$, $\beta(u) = \alpha^{-1}(\beta(\alpha_i(u)))$.\\}.
Let us now construct $C^+ = (G_1^+ = G^+, G_2^+,\ldots)$; see also Figure~\ref{fig:C+} for an illustration:

\begin{enumerate}
\item $C^+_i := (G_1^+,\ldots, G_{i-1}^+)$ is the extension of $C'_{<i}$ to $G^+$, i.e., the same contractions are performed, ignoring $H$. Note that this construction shows that $G_j' \uparrow G^+$ for each $j < i$, as required.
\item $C^+_j := (G_{i-1}^+, \ldots, G_j^+)$ is the extension of $C_H$ to $G_{i-1}^+$, i.e., $C_H$ is applied to $H$, ignoring the rest of $G_{i-1}^+$.
\item Let $H_j$ and $H'_j$ be the subtrigraphs of $G_j^+$ induced by the descendants of $H$ and $H'$, respectively. 
By Definition of $C_H$, there is a bijection $\alpha_j$ from $V(H_j)$ to $V(H_j')$ that respects\footnote{By respecting $\iota$, we mean that if $u \in \beta(v)$ for $u \in V(H)$, $v \in V(H_j)$, then $\iota(u) \in \beta(\alpha_j(v))$.} $\iota$. Let $C_k^+ := (G_j^+, \ldots, G_k^+)$ be the contraction sequence that contracts $u$ and $\alpha_j(u)$ for every $u \in V(H_j)$ in arbitrary order.
\item We will prove that $G_k^+ \cong G_{i-1}'$, and we will define the rest of $C^+$ to be the suffix of $C'$ starting with $G_i'$.
\end{enumerate}

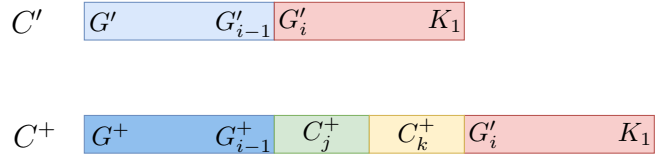
\begin{figure}[t]
\begin{tikzpicture}[y=0.5cm, x=0.5cm]
\definecolor{c6c8ebf}{RGB}{108,142,191}
\definecolor{cdae8fc}{RGB}{218,232,252}
\definecolor{cb85450}{RGB}{184,84,80}
\definecolor{cf8cecc}{RGB}{248,206,204}
\definecolor{c90bfee}{RGB}{144,191,238}
\definecolor{c82b366}{RGB}{130,179,102}
\definecolor{cd5e8d4}{RGB}{213,232,212}
\definecolor{cd6b656}{RGB}{214,182,86}
\definecolor{cfff2cc}{RGB}{255,242,204}

\clip(-2.5,-0.5) rectangle (16,4.5);
  \path[fill=white];
  
     \path[draw=c6c8ebf,fill=cdae8fc] (0.0, 4) rectangle (5, 3);
  
     \path[draw=cb85450,fill=cf8cecc] (5, 4) rectangle (10, 3);
  
     \path[draw=c6c8ebf,fill=c90bfee] (0, 1) rectangle (5, 0.);
  
     \path[draw=cb85450,fill=cf8cecc] (10, 1) rectangle (15, 0.);
  
     \path[draw=c82b366,fill=cd5e8d4] (5, 1) rectangle (7.5, 0);
  
     \path[draw=cd6b656,fill=cfff2cc] (7.5, 1.) rectangle (10, 0.);
  \node at (-1.5, 3.5) {\Large $C'$};
  \node at (-1.3, 0.5) {\Large $C^+$};
  \node at (5.5, 3.5) {$G_i'$};
  \node at (4.2, 3.4) {$G_{i-1}'$};
  \node at (0.5, 3.5) {$G'$};
  \node at (9.5, 3.5) {$K_1$};
  \node at (0.7, 0.5) {$G^+$};
  \node at (4.2, 0.4) {$G_{i-1}^+$};
  \node at (10.5, 0.5) {$G_i'$};
  \node at (14.5, 0.5) {$K_1$};
  \node at (6.25, 0.5) {$C_j^+$};
  \node at (8.75, 0.5) {$C_k^+$};
\end{tikzpicture}
\caption{A schematic depiction of the construction of $C^+$ from $C'$ in the proof of Lemma~\ref{lem:one_new_H}. Informally, we insert a new contraction segment after $G'_{i-1}$ (the green and the yellow block), which handles $H$. The blue prefixes of the two contraction sequences are ``morally'' the same but in $C^+$, $H$ is still present, and so $G_\ell'$ is not isomorphic to $G^+_\ell$ for $\ell \in [i-1]$ but it is an induced subtrigraph thereof. On the other hand, the red suffixes are exactly the same since $H$ has been contracted with $H'$. 
\label{fig:C+}}
\end{figure}

Let us argue that $C^+$ can be computed in polynomial time. First, we need to find the two merged graphs $H', H'' \in \ca H$: this can be done by brute force because the size of $\ca H$ is at most $\ca O(n)$ and checking whether given $H'$ and $H''$ are merged can be done efficiently (the details depend on the computational model and the representation of contraction sequences). Then, we compute $C_H$ by going through $C'_{<i}$ and looking only at contractions involving vertices of $H'$. Using $C_H$, it is easy to compute $C_j^+$. All other parts of $C^+$ can clearly be computed in polynomial time.

Now we need to show that $C^+$ has progressive width $(t \rightarrow_i 2t)$. By the assumption about the progressive width of $C'$, $C^+_i$ has width at most $t$ (using also the fact that $G_{i-1}'$ is a $C'$-safe trigraph for $H$; no red edge in $C^+_i$ is incident to $H$). Hence, we only need to prove that the suffix of $C^+$ starting with $G_i^+$ has width at most $2t$.
Let $S_j^H$ be the set containing the descendants $S^H$ in $G_j^+$ (or, equivalently, in $G_{i-1}'$, $G_{i-1}^+$ or $G_k^+$).

\begin{claim}\label{cl:Cj}
$C_j^+$ has width at most $2t$. Moreover, descendants of $H$ have red degree at most $t$ in trigraphs of $C_j^+$.
\end{claim}
\begin{claimproof}[Proof of the Claim]
Let $\ell \in [i, j]$, let $H_\ell$ be the subtrigraph of $G_\ell^+$ induced by the descendants of $H$, and let $m \in [i-1]$ be an index such that the subtrigraph $H'_{m}$ of $G_{m}^+$ induced by the descendants of $H'$ satisfies $|V(H_\ell)| = |V(H'_m)|$. We need to show that the red degree of each $u \in V(G^+_\ell)$ is at most $2t$ (and at most $t$ when $u \in V(H_\ell)$). By construction of $C_j^+$, there is a bijection $\alpha: V(H_\ell) \rightarrow V(H'_m)$ such that if $u \in \beta(v)$ for $u \in V(H)$, $v \in V(H_\ell)$, then $\iota(u) \in \beta(\alpha(v))$.

Let $u \in V(H_\ell)$.
We will construct a (partial) injection $\gamma: V(G^+_\ell) \rightarrow V(G^+_m)$ such that if $uv \in R(G^+_\ell)$, then $\alpha(u)\gamma(v) \in R(G^+_m)$. Since $\alpha(u)$ has red degree at most $t$ in $G^+_m$, this will prove that $u$ has red degree at most $t$ in $G^+_\ell$.
Let $v \in V(G^+_\ell)$ be a red neighbor of $u$ in $G^+_\ell$.
There are two cases to be considered:

\begin{enumerate}
\item If $v \in V(H_\ell)$, then $\alpha(u)\alpha(v) \in R(H'_m)$, using the fact that $\beta(u), \beta(v) \seq V(H)$, and we set $\gamma(v) := \alpha(v)$.
\item If $v \notin V(H_\ell)$, then $v \in S^H_j$ by construction of $C_j^+$. Let $v_0 \in \beta(v)$. By Observation~\ref{obs:subset}, $\beta(v)$ is a subset of an equivalence class of $\sim_H$. Hence, there are $u_0, u_1 \in \beta(u)$ such that $u_0v_0 \in E(G)$ but $u_1v_0 \notin E(G)$, and we let $\gamma(v)\in V(G_m^+)$ be the unique vertex such that $v_0 \in \beta(\gamma(v)) \seq \beta(v)$.
\end{enumerate}

Now we only need to show that a vertex $v \in S^H_j$ has red degree at most $2t$ in $G^+_\ell$ (no other vertex is affected by contractions among descendants of $H$). Let $K \seq V(H_\ell)$ be the set of red neighbors of $v$ in $H_\ell$ (in $G^+_\ell$). 
By Observation~\ref{obs:subset}, some (actually, each) ancestor $v_0 \in V(G^+_m)$ of $v$ has among its red neighbors all vertices of $\alpha(K)$ in $G^+_m$. Since the red degree of $v_0$ is at most $t$ in $G^+_m$ and $\alpha$ is a bijection, we obtain that $|K| \le t$. Hence, $v$ has at most $t$ red neighbors in $H_\ell$ (in $G_\ell^+$). All other red neighbors of $v$ in $G_\ell^+$ are its red neighbors also in $G_{i-1}^+$ (which has maximum red degree at most $t$), and so $v$ has indeed red degree at most $2t$ in $G^+_\ell$.
\end{claimproof}

\begin{claim}\label{cl:Ck}
$C_k^+$ has width at most $2t$.
\end{claim}
\begin{claimproof}[Proof of the Claim]
Let $\ell \in [j, k]$, let $H_\ell, H_\ell'$ be subtrigraphs of $G_\ell^+$ induced by the descendants of $H$ and $H'$, respectively, let $H_\ell^+ := H_\ell \cup H_\ell'$
, and let $\alpha_j\colon V(H_j) \rightarrow V(H_j')$ be the bijection defined in the construction of $C^+_k$. We need to show that the maximum red degree in $G^+_\ell$ is at most $2t$.

First, let $v \in V(G_\ell^+ - H_\ell^+)$. By construction of $C_k^+$, we know that $v \in V(G_j^+)$. Suppose that $v$ has higher red degree in $G^+_\ell$ than in $G^+_j$. This can happen only if a black edge $uv \in E(G_j^+)$ becomes red because of a contraction involving $u$. However, the only contractions happening in $C_k^+$ are between $u$ and $\alpha_j(u)$ for some $u \in V(H_j)$, and $uv \in E(G^+_j)$ if and only if $\alpha_j(u)v \in E(G^+_j)$, by definition of $\alpha_j$.
Hence, the red degree of $v$ in $G^+_\ell$ is at most its red degree in $G^+_j$, and that is at most $2t$ by Claim~\ref{cl:Cj}.

Second, we need to show that each $u \in V(H_\ell^+)$ has red degree at most $2t$ in $G_\ell^+$.
Observe that $H_\ell'$ contains no black edges because each vertex $u \in V(H_\ell')$ is a descendant of both $H'$ and $H''$. Hence, a vertex $u \in V(H_\ell') \setminus V(H_\ell)$ has red degree at most $t$ in $G^+_\ell$ because it cannot have higher red degree in $G^+_\ell$ than in $G^+_j$. Conversely, let $u \in V(H_\ell)$ and let $d$ be the degree of the ancestor $u_0 \in V(H_j)$ of $u$ in $H_j$. Observe that $u$ has degree at most $2d$ in $H_\ell^+$: for each neighbor $v_0 \in V(H_j)$ of $u_0$ in $H_j$, $u$ can have two neighbors in $H_\ell^+$, namely $v_0$ and $\alpha_j(v_0)$; this happens when $u_0$ has been contracted with $\alpha_j(u_0)$ into $u$ but no neighbor $v_0 \in V(H_j)$ of $u_0$ has been contracted with $\alpha_j(v_0)$. 
Moreover, $u_0$ and $\alpha_j(u_0)$ have exactly the same red neighbors in $S_j^H$ (by definition of $\alpha_j$).  
Hence, the red degree of $u$ in $G^+_\ell$ has increased by at most $d \le t$, compared to the red degree of $u_0$ in $G_j^+$, and so $u$ has at most $t + d \le 2t$ red neighbors, which concludes the proof.
\end{claimproof}

Since $H_j'$ contains no black edges (each of its vertices is a descendant of both $H'$ and $H''$), the contraction of $H_j$ and $H_j'$ creates no new red edge (using also the fact that $H_j$ and $H_j'$ are attached to $S_j^H$ in the same way). Hence, we obtain that 
$G_k^+ \cong G_j^+ - H_j \cong G_{i-1}'$, and we can indeed define the rest of $C^+$ to be the suffix of $C'$ starting with $G_i'$. This suffix has width at most $2t$, since $C'$ has progressive width $(t \rightarrow_{i} 2t)$.
\end{proof}

Now we are finally ready to prove Lemma~\ref{lem:CfromC'}. This is the only remaining part of this section because we have already shown how Lemma~\ref{lem:CfromC'} implies Theorem~\ref{thm:param-by-VI}, see Subsection~\ref{sub:5setup}.

\begin{proof}[Proof of Lemma~\ref{lem:CfromC'}] 
The idea of the proof is to iteratively apply Lemma~\ref{lem:one_new_H} to all the graphs in $\ca C$ not present in the reduced graph $G'$. However, this requires some care, as applying the lemma 
in the wrong order might fail to ensure the precondition on the progressive-width.
In order to prove this lemma, we will consider the following key claim:

\begin{claim}\label{cl:ind_update}

Given $G^*$, $C^*$, and $\mathcal{L^*}$ satisfying the following properties, we can construct in polynomial time a contraction sequence $C$ of width at most $2t$ for $G$.
\begin{enumerate}
\item $G^*$ is a $\con$-respecting graph;
\item $\mathcal{L^*}$ is a list of pairs (graph $H$, integer $\delta$), such that the integer value is non-increasing;
\item Each pair $(H,\delta)$ in $\mathcal{L^*}$ satisfies all of the following: (i.) $H \in \con$, and $H$ appears only once in $\mathcal{L^*}$, 
(ii.) $H \nsubseteq G^*$, 
(iii.) $[H]_{\sim}$ is large 
in $G^*$,
(iv.) $G^*_{\delta}$ is $C^*$-safe for $H$;
\item $C^*$ is a contraction sequence for $G^*$ of width at most $2t$, and if $(H_0, \delta_0)$ is the first pair in $\ca L^*$, then $C^*$ has progressive width $(t\rightarrow_{\delta_0+1} 2t)$;
\item $V(G^*)\cup \bigcup_{(H,\delta)\in\mathcal{L^*}} V(H)=V(G)$.
\end{enumerate}
\end{claim}

\begin{claimproof}[Proof of Claim~\ref{cl:ind_update}]
We proceed by induction on the length of $\mathcal{L^*}$.
The base case is trivial: if $\mathcal{L^*}$ is the empty list, the conditions 1. and 5. ensure that $G^*=G$, and 4. ensures that $C^*$ has width $2t$.

Now let us suppose that the claim is true for any list of length $i$, for some $i\geq 0$. Consider $G^*$, $C^*$, $\mathcal{L^*}$ satisfying the hypothesis such that $\mathcal{L^*}$ contains $i+1$ elements, the first of which being $(H_0, \delta_0)$. We can apply Lemma~\ref{lem:one_new_H} to $G^*$, $C^*$ and $H_0$ since the points 1., 3. and 4. are exactly the preconditions of the lemma, and we obtain in polynomial time a contraction sequence $C^+$ of progressive width $(t\rightarrow_{\delta_0+1} 2t)$ for $G^+=G[V(G^*)\cup V(H_0)]$.

Now let us consider $\mathcal{L^+}$ the suffix of $\mathcal{L^*}$ of length $i$---i.e., we only remove $(H_0, \delta_0)$---and prove that $G^+$, $C^+$, and $\mathcal{L^+}$ satisfy the requirements to apply the induction hypothesis.

The first obvious point is that the length of $\mathcal{L^+}$ is $i$. Since $G^*$ is $\con$-respecting and $H\in \con$, we obtain that $G^+$ is $\con$-respecting, i.e., it satisfies 1. We can easily verify 5.: \[V(G^+)\cup \bigcup_{(H,\delta)\in\mathcal{L^+}} V(H)=V(G^*)\cup V(H_0) \bigcup_{(H,\delta)\in\mathcal{L^+}} V(H)=V(G^*)\cup \bigcup_{(H,\delta)\in\mathcal{L^*}} V(H) =
V(G).\]

As a suffix of $\mathcal{L^*}$, $\mathcal{L^+}$ satisfies 2., and the first three requirements in 3. are also trivially satisfied. To prove the 3.iv., it is necessary to observe two things. First, observe that for each pair $(H, \delta) \in \mathcal{L^+}$, it holds that $G^+_{\delta}=G^*_{\delta}\uparrow G^+$ since $\delta \leq \delta_0$, by Lemma~\ref{lem:one_new_H} (the ``moreover'' part).
Second, observe that there is no red edge in $G^+_\delta$ that is not already present in $G^*_{\delta}$: indeed, any such red edge would be incident to $H_0$ by construction of $G^+_\delta$, and its existence would contradict the definition of $\delta_0$, i.e., the $C^*$-safeness for $H_0$ of $G^*_{\delta_0}$. Hence we conclude that for every $(H, \delta)\in \mathcal{L^+}$, it holds that $G^+_{\delta}$ is $C^+$-safe for $H$.

The last item to check, requirement 4., is easily handled: we know that $C^+$ has progressive width $(t\rightarrow_{\delta_0+1} 2t)$ by Lemma~\ref{lem:one_new_H}, and for all $(H,\delta)\in \mathcal{L^+}$, it holds that $\delta\leq \delta_0$ by 2., i.e., by the monotony of $\mathcal{L^*}$ in the second component.

Using the induction hypothesis, we can now create in polynomial time a contraction sequence $C$ of width at most $2t$ for $G$. The total running time is polynomial, hence the claim is proven.
\end{claimproof}

To finish the proof of Lemma~\ref{lem:CfromC'}, we only need to construct the initial list $\ca L'$ for $G'$. For each graph $H\in \con$ such that $H\nsubseteq G'$, let $\delta(H)$ be the index of the last $C'$-safe trigraph for $H$, whose existence is ensured by Lemma~\ref{lem:safe-exists}.
Let $\mathcal{L'}$ be the list of pairs $(H, \delta(H))$, ordered by non-increasing values of $\delta(H)$, and recall that $C'$ is given.
It is easy to see that the requirements on $G'$, $C'$ and $\mathcal{L'}$ are satisfied---either by definition or by construction---to apply Claim~\ref{cl:ind_update}: we obtain in polynomial time a contraction sequence $C$ of width at most $2t$ for $G$, and since the creation of $\mathcal{L'}$ can be achieved in polynomial time, we have proven the lemma.
\end{proof}

\section{Concluding Remarks}
While we believe that the results presented here provide an important contribution to the state of the art in the area of computing twin-width, many prominent questions still remain unanswered. Apart from the ``grand prize''---resolving the parameterized approximability of twin-width when the runtime parameter is twin-width itself---future research may focus on finding fixed-parameter algorithms that compute optimal or near-optimal contraction sequences under less restrictive runtime parameters than those considered in this article. 

More specifically, the problem remains entirely open when parameterized by treewidth and treedepth, and resolving this may require new insights into the structural properties of optimal contraction sequences and lead to tighter bounds on the twin-width of well-structured graphs. For treedepth in particular, we suspect that combining the ideas presented in Section~\ref{sec:vi} with the \emph{iterative pruning} approach typically used for treedepth-based algorithms~\cite{GanianO18,GanianPSS20,BhoreGMN22} may be an enticing direction to pursue; however, we note that such a combination does not seem straightforward.

\bibliographystyle{plainurl}
\bibliography{main.bib}
\end{document}